\documentclass[sn-mathphys-ay]{sn-jnl}

\usepackage{graphicx}%
\usepackage{multirow}%
\usepackage{amsmath,amssymb,amsfonts}%
\usepackage{amsthm}%
\usepackage{mathrsfs}%
\usepackage[title]{appendix}%
\usepackage{xcolor}%
\usepackage{textcomp}%
\usepackage{manyfoot}%
\usepackage{booktabs}%
\usepackage{algorithm}%
\usepackage{algorithmicx}%
\usepackage{algpseudocode}%
\usepackage{listings}%
\usepackage{comment}

\theoremstyle{thmstyleone}%
\newtheorem{theorem}{Theorem}
\theoremstyle{thmstyletwo}%
\newtheorem{corollary}{Corollary}%
\theoremstyle{thmstylethree}%

\newcommand{\cent}{\mbox{\textcent}}

\newcommand{\mymatrix}[2]{\left( \begin{array}{#1} #2\end{array} \right)}

\newcommand{\myvector}[1]{\mymatrix{c}{#1}}
\newcommand{\myrvector}[1]{\mymatrix{r}{#1}}

\newcommand{\mypar}[1]{\left( #1 \right) }

\newcommand{\dollar}{\$}

\begin{document}

\title[Article Title]{Two-way affine automata can verify every language}


\author*[1]{\fnm{Zeyu} \sur{Chen}}\email{chenzeyu@zju.edu.cn}

\author[2]{\fnm{Abuzer} \sur{Yakary{\i}lmaz}}\email{abuzer.yakaryilmaz@lu.lv}


\affil[1]{School of Mathematical Sciences, Zhejiang University, Hangzhou 310058, People's Republic of China}

\affil[2]{\orgdiv{Center for Quantum Computer Science}, \orgname{University of Latvia}, \orgaddress{\street{Street}, \city{R\={\i}ga}, \postcode{LV-1586}, \country{Latvia}}}



\abstract{When used as verifiers in Arthur-Merlin systems, two-way quantum finite automata can verify membership in all languages with bounded error with double-exponential expected running time, which cannot be achieved by their classical counterparts. We obtain the same result for affine automata with single-exponential expected time. We show that every binary (and r-ary) language is verified by some two-way affine finite automata verifiers by presenting two protocols: A weak verification protocol uses a single affine register and the input is read once; and, a strong verification protocol uses two affine registers. These results reflects the remarkable verification capabilities of affine finite automata.}




\maketitle

\section{Introduction}\label{sec1}
Quantum computation promises significant computational advantages in many practical tasks \citep{shor1994,grover1997,peruzzo2014}. In the case of computation with constant memory, it was shown that two-way quantum finite automata are more efficient and more powerful than their classical counterparts \citep{freivalds1981,kondacs1997, AmbainisW02}; and, they can solve certain problems using exponentially less memory when reading the input in realtime mode \citep{ambainis1998}. Here interference plays a pivotal rule in quantum advantage  \citep{hirvensalo2018,AY21}. To further understand the power of interference, \cite{diaz2016} introduced affine computation and affine finite automata (AfAs), simulating quantum interference classically and ``measuring'' the probability based on the $\textit{l}_1$-norm. In recent years, it has been shown that AfAs and their generalizations are more powerful than QFAs and PFAs in various scenarios \citep{villagra2018,nakanishi2017,ibrahimov2018}.

Arthur-Merlin (AM) proof system was first introduced by \cite{babai1985}. It is an interactive proof system in which the verifier's coin tosses (random choices) are visible to the prover. \cite{say2017} showed that two-way finite automata with quantum and classical states (2QCFAs) can verify membership on any language with bounded error as verifiers in AM proof systems, far outperforming their classical counterparts. \cite{khadieva2021} presented a protocol, showing that affine automata with deterministic and affine states (ADfA) have the same power for unary language. In this paper, we extend this result for binary (and $r$-ary) languages by allowing the input tape head to stay on the same symbol or move to the left. 

We define two-way affine finite automaton as a two-way automaton with deterministic and affine states. As a verifier, we add the capability of nondeterministic choices. Then, 
we present a weak protocol (i.e., the non-members does not need to be rejected with high probability), where the verifier uses a single affine register weighted once and never moves its input head to the left. We make this protocol stronger by using an additional affine register executing a probabilistic procedure. The first protocol runs in exponential time for the member strings, and it may run forever in some branches for non-member strings. Then, we ensure that the second protocol terminates in exponential expected time for the non-member strings. Compared to the techniques used in \cite{say2017}  and \cite{khadieva2021}, we present partly different techniques in this paper.
 
The paper is organized as follows. Section \ref{sec:pre} presents basic notations and definitions. Section \ref{sec:2ADfA} introduces two-way automaton with deterministic and affine states, and Section \ref{sec:2ANfA} introduces its nondeterministic version. We present our weak protocol in Section \ref{sec:weak-binary} and our strong protocol in  Section \ref{sec:strong-verification}. We conclude our paper with Section \ref{sec:con}.

\section{Preliminary}
\label{sec:pre}

We assume that the reader is familiar with automata theory, interactive proof systems, and quantum automata. We refer the reader to \cite{Sip13} for the basics of automata theory, \cite{Con93} for an excellent review of space-bounded interactive proof systems, \cite{SY14} for a pedagogical introduction to quantum automata, and \cite{AY21} for a comprehensive review on quantum automata.

\subsection{Common notation}

Throughout the paper, $|\cdot|$ refers to the $\ell_1$-norm; $\Sigma$ denotes the input alphabet not containing $\cent$ and $\$ $ (the left and right end-markers, respectively); $\tilde{\Sigma}$ is the set $\Sigma \cup\{\cent, \$\}$; $\Sigma^*$ denotes the set of all strings defined on the alphabet $\Sigma$ including the empty string denoted $\varepsilon$; and, for a given string $w \in \Sigma^*, \tilde{w}$ denotes the string $\cent w \$$. Moreover, for any string $w,|w|$ is the length of $w$, and for non-empty string $w$, $w_i$ is its symbol from the left. For a given vector $v$, $v[i]$ is its $i$th entry.

\subsection{Affine computation}

An $m$-state affine register $\{e_1,e_2,\ldots,e_m\}$ is represented by $\mathbb{R}^m$, and the affine state of this system is represented by $m$-dimensional vector:

$$
v=\left(\begin{array}{c}
v_1 \\
\vdots \\
v_m
\end{array}\right) \in \mathbb{R}^m
$$
satisfying that $\sum_{j=1}^m v[j]=1$, where $v[j]$ is the value of the system being in the $j$th state. We also use vector 
\[
e_j = \myvector{0 \\ \vdots \\ 0 \\ 1 \\ 0 \\ \vdots \\ 0} \leftarrow j\text{th} \ \text { entry }
\]
to represent that the affine state is in $e_j$.

Any affine operator of this system is a linear operator preserving the entry summation. Formally, an affine operator is represented by an $(m \times m)$-dimensional matrix:
$$
A=\left(\begin{array}{ccc}
a_{1,1} & \cdots & a_{1, m} \\
\vdots & \ddots & \vdots \\
a_{m, 1} & \cdots & a_{m, m}
\end{array}\right) \in \mathbb{R}^{m \times m}
$$
satisfying that $\sum_{j=1}^m |a_{j, i}|=1$ for each column $i$. When the $A$ is applied to the affine state $v$, the new state  
\[     
    v^{\prime}=M v
\]
is also an affine state.
It is trivial to check that the multiplication of two affine operators is an affine operator.

To retrieve information from the register when in $v$, we apply a weighting operator: the $j$th state is observed with probability
\[
P_j = \frac{\left|v_j\right|}{\left|v\right|} \in[0,1].
\] 
If it is observed, the new affine state becomes $e_j$.

\subsection{Discussion about partial weighting}

Weighting is similar to quantum measurement on a computational basis: We observe a single state. It is possible that the quantum system can still be in a superposition after a partial measurement. But this is not allowed for affine systems due to the normalization constraint. Suppose that we are in the affine state
\[
\myrvector{1 \\ -1 \\ 1}.
\]
If we are allowed to make "partial weighting" based on $\{e_1,e_2\}$ and $\{e_3\}$, we would get the following state before normalization:
\[
\myrvector{1 \\ -1 \\ 0 } \mbox{ and } \myvector{0 \\ 0 \\ 1}.
\]
Then, it is easy to see that the first vector cannot be an affine vector as we cannot make the entry summation 1 with any normalization factor.

Therefore, we allow to use more than one affine register so that we can still have ``superposition'' in some, while weighting the others.

\subsection{Encoding binary languages as real numbers}

We use a binary encoding technique given in \cite{say2017} with some modifications.
Let $ \Sigma = \{0,1\}$. We enumerate all possible binary strings in shortlex order as:
\[
\begin{array}{ll}
     1 & \varepsilon \\
     2 &  0 \\
     3 &  1 \\
     4 &  00 \\
     5 &  01 \\
     6 &  10 \\
     7 &  11 \\
     8 & 000 \\
     \vdots & \vdots \\
\end{array}.
\]
We denote the $i$th string in this order $\Sigma^*(i)$. Remark that for any given binary string $w$, we have 
\[
w = \Sigma^*( (1w)_2 ),
\]
where $(1w)_2$ is a number in binary.
Let $G_L$ be the characteristic function $L$ as: 
\begin{itemize}
    \item $G_L(w) = 1 $ if $w \in L$ and
    \item $G_L(w)=0$ if $w \notin L$.
\end{itemize}

We encode all membership information of $L$ on a single real number:
\begin{equation}
    \alpha_{L}=\sum_{i=1}^{\infty} \frac{G_{L}\left(\Sigma^*(i)\right)}{6^i}=\frac{G_{L}\left(\Sigma^*(1)\right)}{6}+\frac{G_{L}\left(\Sigma^*(2)\right)}{6^2}+\frac{G_{L}\left(\Sigma^*(3)\right)}{6^3}+\cdots 
\end{equation}
For $j \geq 1$, we define 
\begin{equation}
\label{alj}
    \alpha_L[j] = \frac{G_{L}\left(\Sigma^*(j)\right)}{6}+\frac{G_{L}\left(\Sigma^*(j+1)\right)}{6^2}+\frac{G_{L}\left(\Sigma^*(j+2)\right)}{6^3}+\cdots   .
\end{equation}
Remark that:
\begin{itemize}
    \item $\alpha_L = \alpha_L[1]$, 
    \item $\alpha_L[j+1] = 6 \cdot \alpha_L[j] - G_L(\Sigma^*(j)) $, and
    \item $\alpha_L[j]$ is a real number between 0 and $\frac{1}{5}$.
\end{itemize}

\subsection{Guessing the digits of $\alpha_L$}
\label{sec:guessing}

For a given language $L \subseteq \Sigma^* = \{0,1\}^* $, let $g_1, g_2, g_3, g_4, \ldots $ be a sequence of guesses for $G_L(\varepsilon), G_L(0), G_L(1), G_L(00), \ldots$. We start with $\beta_1 = \alpha_L[1]$. Then, we iteratively calculate $\beta_i$ as
\[
\beta_i = 6 \cdot \beta_{i-1} - g_{i-1}.
\]
It is easy to see that if our guesses are correct, then we have 
\[
    \beta_i = \alpha_L[i] .
\]
Suppose that our guesses are correct until $g_i$, but $g_i$ is incorrect. Then, we have $\beta_i = 6 \cdot \alpha_L[i-1] - g_i$.
\begin{itemize}
    \item If $g_i = 0$: $\beta_i = \alpha_L[i] + 1 $.
    \item If $g_i = 1$: $\beta_i = \alpha_L[i] - 1 $.
\end{itemize}
We observe that $0 \leq |\beta_j| \leq \frac{1}{5}$ for $j < i$, i.e., it is bounded above by $\frac{1}{5}$. Now, we have $|\beta_i| \geq \frac{4}{5} $, i.e., it is bounded below by $\frac{4}{5}$. In the next iterations, independent of correctness of next guesses, the bound from below will get bigger and bigger. Thus, we can conclude that
\begin{equation}
\label{eq:incorrect-guess}
    |\beta_j| \geq \frac{4}{5} \mbox{ for } j \geq i, 
\end{equation}
where $g_i$ is the first incorrect guess in the sequence. We refer to this fact later in our proofs.

\section{Two-way affine finite automaton}
\label{sec:2ADfA}

We start with defining two-way deterministic finite automaton (2DFA). An $n$-state 2DFA $M$ is a 6 tuple
\[
    M = (S,\Sigma,\delta,s_I,s_a,s_r),
\]
where
\begin{itemize}
    \item $ S = \{s_1,\ldots,s_n\} $ is the set of states;
    \item $\delta : S \times \tilde{\Sigma} \rightarrow S \times \{-1,0,1\}$ is the transition function described below; and,
    \item $s_I \in S$, $s_a \in S$, and $s_r \in S$ are the initial, accepting, and rejecting states, respectively, with $s_I \neq s_a \neq s_r$.
\end{itemize}

The automaton $M$ has a read-only input tape with a reading head. The tape is one-way infinite and the tape cells (squares) are indexed by $0, 1, 2, \ldots $. For a given input $w \in \Sigma^*$, the string $\tilde{w} =\cent w \dollar $ is placed on the input tape between the $0$th and $(|w|+1)$th cells. At the beginning of computation, $M$ is in $s_I$ and the input head is placed on $\cent$. The computation is governed by the transition function $\delta$. When in state $s \in \Sigma$ and reading symbol $\sigma \in \tilde{\Sigma}$, the automaton $M$ switches to state $s' \in S$ and the head position is updated by $d \in \{-1,-0,1\}$, where $\delta(s,\sigma) = (s',d)$. Here $-1$ and $1$ means the head is moved one cell to the left and the right, respectively, and, $0$ means the head is stationary. It must be guaranteed that the automaton never attempts to leave $\tilde{w}$. If $M$ enters state $s_a$ (resp., $s_r$), then the computation is terminated and the input is accepted (resp., rejected). 

We define two-way affine automaton as a 2DFA augmented with $k>0$ affine registers. Formally, a two-way automaton with deterministic and affine states (2ADfA) $M$ is a 7-tuple
\[
M = (S,\Sigma,\delta, s_I, s_a, s_r, \{R_1,R_2,\ldots,R_k\}),
\]
where, different from a 2DFA, $R_i$ is the $i$th affine register and $\delta$ is composed by two transition functions described below. 

The register $R_i$ is a double
\[
    R_i = ( E_i, \mathcal{A}_i = \{ A_{i,1},\ldots,A_{i,l_i} \} ),
\]
where $E_i = \{e_1,\ldots,e_{m_i}\}$ is the set of affine states and $\mathcal{A}_i$ is the set of affine operators applied to the $i$th affine register. We fix $e_1$ as the initial state for each register.

The computation of the 2ADfA $M$ is governed classically. In each step, we have two phases: affine and classical. By using the notation for  2DFA given above, in the affine phase, based on $(s,\sigma)$, $M$ applies either an affine operator or weighting operator to each affine register. It is possible that some registers are weighted, and some others are updated by affine operators. If any register is weighted, then one of its states is observed and the classical transition takes into account this outcome. In other words, $(s',d)$ is determined by not only $(s,\sigma)$ but also the observed affine state(s) (if there is any). The termination and decision are made classically. 

Formally, $\delta = (\delta_a,\delta_c)$, where $\delta_a$ governs the affine phase and $\delta_c$ governs the following classical phase. For each $(s,\sigma) \in S \times \tilde{\Sigma} $,
\begin{equation}
    \label{eq:delta_a}
    \delta_a(s,\sigma) = (O_1,\ldots,O_k),
\end{equation}
where $O_i \in \mathcal{A}_{i} \cup \{ W_i \}$ is the operator applied to the $i$the affine register with the exception that $R_i$ is weighted if $O_i = W_i$. If $R_i$ is weighted, then the outcome is in $ \{1, \ldots, m_i\} $, the indices of the affine states. If it is not  (i.e., an affine operator is applied), then we assume that the outcome is 0. 

After the affine phase, $\delta_c$ makes a classical transition as:
\begin{equation}
    \label{eq:delta_c}
    \delta_c(s,\sigma, \tau_1,\ldots,\tau_k ) = (s',d),
\end{equation}
where $\tau_i = 0 $ if an affine operator is applied during the affine phase, and $\tau_i \in \{1,\ldots,m_i\}$ is the observed outcome if $R_i$ is weighted during the affine phase.

Remark that after each weighting operator, the computation can split into branches with some probabilities. So, the computation of a 2ADfA on $w$ form a computation tree, where each node is 
\[
(s,j,v_1,\ldots,v_k),
\]
where $s$ is the current state, $j$ is the position of input head, and $v_i$ is the affine state of $R_i$. After applying $\delta$, this node creates one or more children. Here each edge to a child can be assigned with its nonzero occurring probability. The root of this tree is 
\[
(s_I,0,e_1,\ldots,e_1).
\]
It is possible that this tree can be infinite. If there is no branching, then we can see the 2ADfA $M$ as a 2DFA. Each leave node is either accepting or rejecting. (It is possible that the number of leaves can be infinite.) Based on the leaves, we calculate the total accepting and rejecting probability of $M$ on $w$, denoted $Acc_M(w)$ and $Rej_M(w)$, with $ 0 \leq Acc_M(w) + Rej_M(w) \leq 1  $.

\section{Two-way affine finite automata verifiers}
\label{sec:2ANfA}

We focus on only Arthur-Merlin proof systems in this paper. As introduced  by \cite{CHPW98} and lated used by \cite{khadieva2021}, we define our two-way affine verifiers as a nondeterministic version of 2ADfA given in the previous section.

A two-way finite automata with nondeterministic and affine states (2ANfA) $N$ is the same tuple of a 2ADfA
\[
N = (S,\Sigma,\delta, s_I, s_a, s_r, \{R_1,R_2,\ldots,R_k\})
\]
except that the transition functions can make nondeterministic choices. Here both $\delta_a$ and $\delta_c$ can make nondeterministic choices. Instead of Eq.~\ref{eq:delta_a}, we have
\begin{equation}
    \label{eq:non_delta_a}
    \delta_a(s,\sigma) \rightarrow \{ (O_{1,1},\ldots,O_{1,k}),(O_{2,1},\ldots,O_{2,k}), \ldots, (O_{t,1},\ldots,O_{t,k}) \} .
\end{equation}
Thus, $N$ nondeterministically picks one of the defined $t$ options and then implements it.
For the classical phase, instead of Eq.~\ref{eq:delta_c}, we have
\begin{equation}
    \label{eq:non_delta_c}
    \delta_c(s,\sigma, \tau_1,\ldots,\tau_k ) \rightarrow  \{ (s_1',d_1), (s_2',d_2), \ldots, (s_z',d_z)  \} .
\end{equation}
Thus, $N$ nondeterministically picks one of the defined $z$ options and then implements it. Of course, due to the nondeterministic nature of $\delta_a$, we may have different $(\tau_1,\ldots,\tau_k )$'s in the classical phase and so different $z$ values. For each nondeterministic rule, the values of $d$ and $z$ may be different.

For each nondeterministic strategy (by picking one of the available options)\footnote{Here each nondeterministic strategy corresponds to the communication with a prover.}, we have particular $Acc_N(w)$ and $Rej_N(w)$. 

A language $L \subseteq \Sigma^*$ is said to be weakly verified by $N$ with error bound $\epsilon < 1/2$ if and only if we satisfy the following two conditions.
\begin{itemize}
    \item When $ x \in L$: there is a nondeterministic strategy such that $Acc_N(w) \geq 1 - \epsilon $.
    \item When $x \notin L$: for any nondeterministic strategy we have $ Acc_N(w) \leq \epsilon  $. 
\end{itemize}
Here it is not required to reject the non-member inputs with high probability. 

A language $L \subseteq \Sigma^*$ is said to be strongly verified by $N$ with error bound $\epsilon < 1/2$ if and only if we satisfy the following two conditions.
\begin{itemize}
    \item When $ x \in L$: there is a nondeterministic strategy such that $Acc_N(w) \geq 1 - \epsilon $.
    \item When $x \notin L$: for any nondeterministic strategy we have $ Rej_N(w) \geq 1 - \epsilon  $. 
\end{itemize}

In these cases, it is also said that (i) $L$ is weakly/strongly verified by $N$ with bounded error or (ii) bounded-error $N$ weakly/strongly verifies $L$.

\section{Weak verification of every binary language}
\label{sec:weak-binary}

We start with a weak verification protocol where the input is read once and a single affine register is used and weighted once.

\begin{theorem}
    \label{thm:weak-binary}
    For any given binary language $L \subseteq \{0,1\}^*$, there is a 2ANfA, say $N_1$, weakly verifying it with error bound $1/3$  and never moving its head to the left. The 2ANfA $N_1$ uses a single affine register and it is weighted only once. For the member strings, the verification takes in exponential time. 
\end{theorem}
\begin{proof}
    We use an affine register with 5 states. Let $w \in \{0,1\}^*$ be the given input. Let $ K $ be equal to $1w$ in binary. Note that $\Sigma^*[K] = w$, and the membership bit of $w$ is the $K$th digit of $\alpha_L$.

    By reading $\cent w $ from left to right, we deterministically set the affine state to
    \[
    \myrvector{1 \\ K \\ -K \\ 0 \\ 0 } .
    \]
    Here we apply $A_0$ when reading $0$ and $A_1$ when reading $\cent$ or $1$:
    \begin{equation}
    \label{eq:Add-matrix}
    A_0 = 
    \mymatrix{rrrrr}{
    1 & 0 & 0 & 0 & 0  \\
    0 & 2 & 0 & 0 & 0  \\
    0 & -1 & 1 & 0 & 0  \\
    0 & 0 & 0 & 1 & 0  \\
    0 & 0 & 0 & 0 & 1  \\
    }
    ~~\mbox{ and }~~
    A_1 = 
    \mymatrix{rrrrr}{
    1 & 0 & 0 & 0 & 0  \\
    1 & 2 & 0 & 0 & 0  \\
    -1 & -1 & 1 & 0 & 0  \\
    0 & 0 & 0 & 1 & 0  \\
    0 & 0 & 0 & 0 & 1  \\
    } .
    \end{equation}
    
    On symbol $\dollar$, the tape head becomes stationary. First, the affine state is set to 
    \[
    \myrvector{1 \\ K \\ -K \\ \alpha_L \\ -\alpha_L } = 
    \mymatrix{rrrrr}{
    1 & 0 & 0 & 0 & 0  \\
    0 & 1 & 0 & 0 & 0  \\
    0 & 0 & 1 & 0 & 0  \\
    \alpha_L & 0 & 0 & 1 & 0  \\
    -\alpha_L & 0 & 0 & 0 & 1  \\
    } \myrvector{1 \\ K \\ -K \\ 0 \\ 0 } .
    \]
    After that, we enter an infinite loop. In each iteration ($i=1,2,3,\ldots$):
    \begin{itemize}
        \item Nondeterministically guess $G_L(\Sigma^*[i])$, say $g_i$.
        \item As explained in Section~\ref{sec:guessing}, the fourth (resp., fifth) entry of affine state is multiped by 6 and the guessed value $g_i$ is subtracted (resp., added).
        \item At the same time, the second (resp., third) entry of affine state is decreased (resp., increased) by 1.   
         We use the following affine operators for these updates:
        \[
            A'_g = 
            \mymatrix{rrrrr}{
            1 & 0 & 0 & 0 & 0  \\
            -1 & 1 & 0 & 0 & 0  \\
            1 & 0 & 1 & 0 & 0  \\
            -g & 0 & 0 & 6 & 0  \\
            g & 0 & 0 & -5 & 1  \\
            },
        \]
         where $g \in \{0,1\}$ is the guessed value.
        \item The loop is continued in one branch, and the loop is terminated in another branch.
    \end{itemize}
    The decision is made in the branch existing from the loop. Once existed, the affine state is
    \[
    \myrvector{1 \\ K-i \\ -K+i \\ \beta_i \\ -\beta_i}.
    \]
    And we update the affine state as follows
    \begin{equation}
        \label{eq:betai}
    \myrvector{1 \\ K-i \\ -K+i \\  5\beta_i / 4 \\ -5\beta_i / 4 }
    =
    \mymatrix{rrrrr}{
    1 & 0 & 0 & 0 & 0  \\
    0 & 1 & 0 & 0 & 0  \\
    0 & 0 & 1 & 0 & 0  \\
    0 & 0 & 0 & 5/4 & 0  \\
    0 & 0 & 0 & -1/4 & 1  \\
    }
    \myrvector{1 \\ K-i \\ -K+i \\ \beta_i \\ -\beta_i}.
    \end{equation}
    If $g_i$ is 0, then the input is rejected. If $g_i$ is 1, then we weight the affine state. If the outcome is $e_1$, then $w$ is accepted, and it is rejected, otherwise.

    If $w \in L$, then there is a nondeterministic branch such that the loop is executed $K$ times, the first $K$ digits of $\alpha_L$ is guessed correctly, and $g_K = 1$. The affine register in this case before being weighted is
    \[
    \myrvector{1 \\ 0 \\ 0 \\ 5 \alpha_L[K+1] / 4 \\ -5 \alpha_L[K+1] / 4}.
    \]
    We know that $\alpha_L[K+1]$ can be at most $\frac{1}{5}$. Then, the accepting probability of $w$ is at least
    \[
    \dfrac{1}{1+1/4+1/4} = \dfrac{2}{3},
    \]
    where the contribution of the fourth and fifth entries can be at most $1/4$.

    If $w \notin L$, then there are different cases. 
    \begin{itemize}
        \item If $g_i = 0$, $w$ is rejected with probability 1.
        \item If $g_i = 1$ but $i \neq K$, the absolute values of the second and third entries are at least 1. Then, $w$ is rejected with probability at least $$\frac{2}{2+1} = \frac{2}{3} . $$
        \item If $g_i = 1$ and $i=K$, at least one guess of digits of $\alpha_L$ is incorrect ($g_i$ is definitely incorrect). Then, we know that $| \beta_i |$ is bounded below by $4/5$ in this case. That is, the absolute values of the fourth and fifth entries are at least 1 (see Eq.~\ref{eq:betai}). Thus, $w$ is rejected with probability at least $$\frac{2}{2+1} = \frac{2}{3}.$$
    \end{itemize}

    We calculate the running time of the successful branch for members. Let $l = |w|$  be the length of input. We read $l+2$ symbols before entering the infinite loop. We iterate the loop $1w$ (in binary) times, and this can be at most $ \underbrace{1 ~~ \cdots ~~ 1}_{l+1\mbox{ times}} $ (in binary), that is, $$ = 2^{l+1}-1 . $$ We use two more transition steps to make our decision. Thus, the total number of transition steps can be at most
    \[
        l + 2 + 2^{l+1} - 1 + 2 = 2^{l+1} + l + 3. 
    \]
    So, the running time of the success branch is exponential in $|w|$: $O(2^{|w|})$.
\end{proof}

By using bigger denominator for $\alpha_L$ and tuning the transitions of the affine operator given in Eq.~\ref{eq:betai}, we can reduce the error bound arbitrarily close to 0. For simplicity, let $\frac{1}{k}$ be our error bound, where $k \geq 3$ is an integer. Then, we use $\alpha_{L,k}$ instead of $\alpha_L$:
\[
    \alpha_{L,k} = \sum_{i=1}^{\infty} \frac{G_{L}\left(\Sigma^*(i)\right)}{d^i}=\frac{G_{L}\left(\Sigma^*(1)\right)}{d}+\frac{G_{L}\left(\Sigma^*(2)\right)}{d^2}+\frac{G_{L}\left(\Sigma^*(3)\right)}{d^3}+\cdots ,
\]
where $d = k^2-2k + 3$. (We observe that $d = 6 $ if $k=3$.) We can easily verify that $ 0 \leq \alpha_{L,k}[j] \leq \dfrac{1}{d-1}$ for every $j$. 

Then, we replace Eq.~\ref{eq:betai} as
\begin{equation}
    \label{eq:betai_2}
    \myrvector{1 \\ k(K-i)/2 \\ -k(K+i)/2 \\  c\beta_i \\ -c\beta_i }
    =
    \mymatrix{rrrrr}{
    1 & 0 & 0 & 0 & 0  \\
    0 & k/2 & 0 & 0 & 0  \\
    0 & 1-k/2 & 1 & 0 & 0  \\
    0 & 0 & 0 & c & 0  \\
    0 & 0 & 0 & 1-c & 1  \\
    }
    \myrvector{1 \\ K-i \\ -K+i \\ \beta_i \\ -\beta_i},
\end{equation}
where $c = \dfrac{k^2-2k+2}{2k-2}$.

For a member string, the final affine state is
\[
\myrvector{1 \\ 0 \\ 0 \\ c\beta_i \\ -c\beta_i}.
\]
Then, the accepting probability is
\[
\dfrac{1}{2|c\beta_i|+1}.
\]
The upper bound for $|c\beta_i|$ is 
\[
    \underbrace{\mypar{ \dfrac{k^2-2k+2}{2k-2} }}_{c} 
    \underbrace{ \mypar{ \dfrac{1}{k^2-2k+2} } }_{1/(d-1)} = \dfrac{1}{2k-2}.
\]
By substituting it, the accepting probability $\dfrac{1}{2|c\beta_i|+1}$ is at least
\[
    \dfrac{1}{2 \mypar{ \frac{1}{2k-2}}+1} = \frac{k-1}{k}.
\]

For a non-member string. we have three cases. 
\begin{itemize}
    \item If $g_i=0$, then the input is rejected with probability 1.
    \item If $g_i=1$ but $i \neq K$, then $(K-i)$ is a non-zero integer, and with the contribution of the second and third entries (see Eq.~\ref{eq:betai_2}), the rejecting probability is at least 
    \[
    \dfrac{k/2 + k/2}{ k/2+k/2+1 } = \dfrac{k}{k+1}.
    \]
    \item If $g_i=1$ and $i = K$, the final affine state is
    \[
    \myrvector{1 \\ 0 \\ 0 \\ c\beta_i \\ -c \beta_i}.
    \]
    In such case, $|\beta_i|$ is bounded below by $ \dfrac{d-2}{d-1} $ (easy to be derived from Section~\ref{sec:guessing}). Then, the rejecting probability is 
    \[
    \dfrac{2c|\beta_i|}{2c|\beta_i|+1},
    \]
    which is at least
    \[
    \dfrac{ 2c \mypar{ \frac{d-1}{d-2}} }{ 2c \mypar{\frac{d-1}{d-2}}+1 }.
    \]
    It is easy to see that 
    \[
    2c \mypar{ \dfrac{d-1}{d-2} } 
    = 
    2 \underbrace{\mypar{\dfrac{k^2-2k+2}{2k-2} }}_{c} 
    \underbrace{\mypar{ \dfrac{k^2-2k+1}{k^2-2k+2} }}_{(d-1)/(d-2)} 
    = \dfrac{(k-1)^2}{k-1}
    = k-1.
    \]
    So, the rejecting probability is at least 
    \[
        \dfrac{k-1}{k}.
    \]
\end{itemize}

\begin{corollary}
    \label{cor:weak-binary}
    For any given binary language $L \subseteq \{0,1\}^*$, there is a 2ANfA weakly verifying it with any error bound and never moving its head to the left. It uses a single affine register and it is weighted only once. For the member strings, the verification takes in exponential time.
\end{corollary}

We can extend the results in this section for $r$-ary languages, where $r>2$. Let $w$ be an $r$-ary string with length $l>0$. When ordering only the strings with length $l$, we have
\[
\begin{array}{ll}
    1 & 0 \cdots 0 0 \\
    2 & 0 \cdots 0 1 \\
    3 & 0 \cdots 0 2 \\
    \vdots & \vdots \\
    (w)_r+1 ~ & w \\
    \vdots & \vdots \\
    r^l & 1 \cdots 1 1
\end{array},
\]
where $(w)_r$ is a number in base-$r$. Then, the place of $w$ in the shortlex order is
\[
    (w)_r+1 + \mypar{ r^{l-1} + r^{l-2} + \cdots + r^1 + r^0 },
\]
where the right side is the number of all strings shorter than $w$, i.e., $r^j$ is the number of strings with length $j$. 

In the proof of Theorem~\ref{thm:weak-binary}, the $K$th digit of $\alpha_L$ represents the membership bit of binary string $w$. So, here, if we calculate $K$ correctly for $r$-ary string $w$ after reading $\cent w \dollar$, the rest of the proof will be the same. As already shown in the previous paragraph, 
\[
K = (w_1 \cdot r^{l-1} + w_2 \cdot r^{l-2} + \cdots + w_{l-1} \cdot r^1 + w_l)  + ( r^{l-1}+r^{l-2} + \cdots + r^0 ) + 1 , 
\]
which is rearranged as
\[
    (w_1+1) \cdot r^{l-1} + (w_2+1)\cdot r^{l-2} + \cdots (w_2+1) \cdot r^1 + (w_l+1) \cdot r^0 + 1. 
\]
Encoding by affine operators is straightforward. Our aim is to set the affine state to
\[
\myrvector{1 \\ K \\ -K \\ 0 \\ 0}
\]
before entering the infinite loop. We read $ \cent ~ w_1 ~ w_2 ~ \cdots ~ w_l ~ \dollar$, and we iteratively calculate $K$. We use the following matrices for each symbol of $w$:
\[
    A''_b = 
    \mymatrix{rrrrr}{
    1 & 0 & 0 & 0 & 0  \\
    b+1 & r & 0 & 0 & 0  \\
    -b-1 & 1-r & 1 & 0 & 0  \\
    0 & 0 & 0 & 1 & 0  \\
    0 & 0 & 0 & 0 & 1  \\
    },
\]
where $b=0,1,\ldots,r-1$.
After that, on the $\dollar$ symbol, we add 1 to the second entry and subtract 1 from the third entry. After reading $\cent$, the affine state is
\[
\myvector{1 \\ 0 \\ 0 \\ 0 \\ 0}.
\]
After reading $w$ and then $\dollar$, we have
\[
\myrvector{1 \\ K \\ -K \\ 0 \\ 0 }
\]
just before entering the infinite loop. 

\begin{corollary}
    \label{cor:weak-rary}
    For any given $r$-ary language $L \subseteq \{0,1,\ldots,r-1\}^*$, where $r>1$, there is a 2ANfA weakly verifying it with any error bound and never moving its head to the left. It uses a single affine register and it is weighted only once. For the member strings, the verification takes in exponential time.
\end{corollary}

\section{Strong verification of every language}
\label{sec:strong-verification}

It is possible to terminate the computation in our weak verification protocol by adding a probabilistic procedure: in each iteration of the infinite loop, we reject the input with some tuned exponentially small probability. For this, we use an additional affine register with two states.

\begin{theorem}
    \label{thm:strong-coin}
    For any given binary language $L \subseteq \{0,1\}^*$, there is a 2ANfA, say $N_2$, with two affine registers strongly verifying it with bounded error in exponential expected time.
\end{theorem}
\begin{proof}
    We modify the 2ANfA $N_1$ given in Theorem~\ref{thm:weak-binary}. Let $w \in \{0,1\}^*$ be the given input with length $|w| > 0$. For the empty string, we make decision deterministically.

    In the infinite loop, after one iteration is completed and just before starting the next iteration, we do a probabilistic experiment and reject the input with probability 
    \[
        \frac{1}{64^{|w|+1}}.        
    \]
    We implement this by moving the input head from one end-marker to another end-marker. For each input symbol and one end-marker, we use an affine operator given below:
    \begin{equation}
    \label{eq:CoinPr}
        \myvector{1\\ \frac{1}{64^{|w|+1}} } =
        \mymatrix{cc}{ 1/64 & ~0 \\ 63/64 & ~1 }^{|w|+1} \myvector{1 \\ 0},
    \end{equation}
    where we start in the first affine state and then apply the affine operator $(|w|+1)$ times.
    We weight this register, and  we reject the input if the first affine state is observed. Otherwise, we continue with the next iteration, and we set this affine register to $e_1$.

    It is clear that this probabilistic experiment does not produce any accepting probability. But, it helps us to terminate the computation in exponential expected time. Suppose that the infinite loop is terminated by only this probabilistic experiment. Then, we give an upper bound for the expected time independent of the nondeterministic strategies. 
    
    Each iteration takes less than $3|w|$ steps, and each iteration terminates with probability $p = \frac{1}{64^{|w|+1}}$ and the next iteration starts with probability $1 - p$. Thus, the expected running time is less than
    \[
        (3|w|) p + (6|w|) (1-p) p + (9|w|) (1-p)^2 p + (12|w|)(1-p)^3 p + \cdots,
    \]
    which is equal to
    \begin{equation}
        \label{eq:exp-time-1}
        3p|w| \mypar{ 1 + 2(1-p) + 3 (1-p)^2 + 4(1-p)^3 + \cdots
    }.
    \end{equation}
    We know that, if $x<0$,
    \[
        1 + x + x^2 + x^3 + x^4 +  \cdots = \frac{1}{1-x}.
    \]
    If we take derivatives of both sides, we have
    \[
    1 + 2x+ 3x^2+ 4x^3 + \cdots = \frac{1}{(1-x)^2}. 
    \]
    Thus, Eq.~\ref{eq:exp-time-1} is simplied as
    \[
    3p|w|\frac{1}{(1-(1-p))^2} = 3p|w|\frac{1}{p^2} = 3 p^{-1} |w|.
    \]
    By substituting $p$, we obtain
    \[
        3|w|64^{|w|+1} = 2^{O(|w|)}.
    \]
    Thus, we ensure that any verification strategy runs in at most exponential expected time.

    If $w \notin L$, it is trivial that the accepting probability is not more than $1/3$, and so, the rejecting probability is at least $2/3$. Remark that the computation terminates with probability 1.

    If $w \in L$, the new accepting probability is calculated as follows. The loop is iterated $K$ ($1w$ in binary) times. 
    \begin{itemize}
        \item The 1st iteration starts with probability 1.
        \item The 2nd iteration starts with probability $(1-p)$.
        \item The 3rd iteration starts with probability $ (1-p)(1-p) = (1-p)^2 $.
        \item The 4th iteration starts with probability $ (1-p)^2 (1-p) = (1-p)^3 $.
        \item $ \cdots $ .
        \item The $K$th iteration starts with probability $ (1-p)^{K-1} $.
    \end{itemize}
    Thus, the accepting probability is at least
    \[
        (1-p)^{K-1} \frac{2}{3}.
    \]
    The value of $K$ is between $2^{|w|}$ and $ 2^{|w|+1}-1$. Then, the value of $(1-p)^{K-1}$ is at least 
    
    \[
    \mypar{ 1 - \frac{1}{64^{|w|+1}} }^{2^{|w|+1}-2}.
    \]
    Remark that for the same length of strings, this expression takes its minimum when $K-1=2^{|w|+1}-2$.
    This expression approaches to 1 quickly. The case for the empty string (when in the language) is trivial, as the probabilistic procedure is never executed.  For $|w|=1$, the value of $(1-p)^{K-1}$ is
    \[
        \mypar{ \dfrac{4095}{4096} }^2 > 1 - \dfrac{2}{4096} > 0.9995.
    \] For bigger $|w|$ values, we get values even closer to 1. Therefore, the accepting probability of member inputs is greater than
    \[
        \mypar{0.9995} \dfrac{2}{3}.
    \]
    
    By executing $N_2$ on the given input several times and then take the majority of answers, we can decrease the error bound arbitrarily close to 0. The number of repetition depends on the error bounds, not the input length. Thus, the expected running time is still exponential. 
\end{proof}

\begin{corollary}
    \label{cor:strong-rary}
    For any given $r$-ary language$L \subseteq \{0,1,\ldots,r-1\}^*$, where $r>1$, there is a 2ANfA with two affine registers strongly verifying it with bounded error in exponential expected time.
\end{corollary}

\section{Conclusion}\label{sec:con}

It is known that 2QCFAs can verify every unary language in exponential expected time, and every binary (and $r$-ary) language in double-exponential expected time \citep{say2017}. \cite{khadieva2021} later showed that AfAs can verify every unary language in linear time (by reading the input in realtime mode). Here we show that using two-way affine automata improves the known quantum time complexity exponentially for binary ($r$-ary) languages: we reduce double-exponential expected time to single-exponential time.

First, we formally define two-way affine automata for the first time, which are not a trivial generalization of (realtime) affine finite automata. We then introduce two verification protocols using different techniques. While we could adapt the 2QCFA technique by using a couple of affine registers, we present a weaker protocol that uses a single affine register, weighted once, and never moves the input head left. We use one additional affine register and extend this to a stronger protocol. The additional affine register basically mimics a probabilistic coin. We remark that both of our techniques can be adapted for 2QCFA verifiers but still with double-exponential time complexity.

As a future direction, the verification power of two-way affine finite automata can be investigated by restricting to use only rational-valued transitions.

\section*{Acknowledgments} 

We are grateful to Professor Junde Wu for initiating and supporting the  collaboration between the authors.

Yakary{\i}lmaz was partially supported by the Latvian Quantum Initiative under European Union Recovery and Resilience Facility project no. 2.3.1.1.i.0/1/22/I/CFLA/001.








\bibliography{ref}

\end{document}